\documentclass{article}
\usepackage{charlesmacros}
\usepackage{spconf}
\usepackage{orcidlink}
\usepackage{algorithmic}
\ninept

\begin{document}

\setlength\abovedisplayskip{5pt}
\setlength\belowdisplayskip{5pt}
\setlength\textfloatsep{5pt}
\title{
        A Distributed Adaptive Algorithm for Non-Smooth Spatial Filtering Problems
}

\name{
        Charles Hovine \orcidlink{0000-0001-6657-1066} and Alexander Bertrand \orcidlink{0000-0002-4827-8568}
        \thanks{This project has received funding from the European Research Council (ERC) under the European Union's Horizon 2020 research and innovation programme (grant agreement No. 802895), the FWO (Research Foundation Flanders) for project
                G081722N and from the Flemish Government under the "Onderzoeksprogramma Artifici\"ele Intelligentie (AI) Vlaanderen" programme.
        }
}

\address{\textit{KU Leuven, Department of Electrical Engineering (ESAT)} \\
        \textit{STADIUS Center for Dynamical Systems, Signal Processing and Data Analytics}\\
        \textit{KU Leuven Institute for Artificial Intelligence (Leuven.AI)}
        \\
       Leuven, Belgium \\
\{charles.hovine, alexander.bertrand\}@esat.kuleuven.be}

\maketitle
\begin{abstract}
        Computing the optimal solution to a spatial filtering problems in a Wireless Sensor Network can incur large bandwidth and computational requirements if an approach relying on data centralization is used. The so-called distributed adaptive signal fusion (DASF) algorithm solves this problem by having the nodes collaboratively solve low-dimensional versions of the original optimization problem, relying solely on the exchange of compressed views of the sensor data between the nodes. However, the DASF algorithm has only been shown to converge for filtering problems that can be expressed as smooth optimization problems. In this paper, we explore an extension of the DASF algorithm to a family of non-smooth spatial filtering problems, allowing the addition of non-smooth regularizers to the optimization problem, which could for example be used to perform node selection, and eliminate nodes not contributing to the filter objective, therefore further reducing communication costs.  We provide a convergence proof of the non-smooth DASF algorithm and validate its convergence via simulations in both a static and adaptive setting.  
\end{abstract}

\section{Introduction}
A spatial filtering problem usually consists in finding the linear combination of a set of signals that is optimal with regards to some criterion, and can therefore be expressed as the solution of an optimization problem. Common examples include principal components analysis \cite{Hotelling1933}, canonical correlation analysis \cite{Kettenring1971}, MAX-SNR beamforming, multichannel Wiener filtering \cite{doclo2007frequency} and  common spatial patterns \cite{koles1990spatial}. 

In the case of Wireless Sensor Networks (WSNs), where several sensing nodes communicate via wireless links, signals are often only short-term stationary, with statistics drifting over time. Being able to \emph{adaptively} compute filters therefore becomes an important requirement. The classical approach to computing spatial filters in WSNs consists in designating a particular node as the \emph{fusion center} (FC), which will collect all the raw data and perform the filter computation centrally \cite{Haykin2010}. This approach is however not ideal, as the bandwidth and computational power required at the FC scales poorly with the number of both nodes and signals. Additionally, the FC constitutes a single point of failure, which can be problematic for many deployment scenarios. An alternative approach consists in solving the filtering problem in a distributed fashion, by sharing the work across the sensor nodes. 

The DASF algorithm \cite{musluoglu2022unified_p1} is a framework for solving adaptive spatial filtering problems in a distributed fashion. Instead of sharing their raw observations, the nodes share efficiently crafted compressed views of their sensor data, which are then used to locally solve low-dimensional versions of the original optimization problem at each node. In addition, the sensor signals' statistics are allowed to change during the course of the algorithm, such that the optimal solution can be tracked adaptively.

The convergence and optimality of DASF in the case of filtering problem expressible as smooth optimization problems, has been studied in \cite{musluoglu2022unified_p2}, but the applicability of the algorithm to non-smooth problems is still unknown. In this paper, we show the convergence and optimality of the algorithm for a family of non-smooth, and possibly non-convex optimization problems. In addition to allowing the algorithm to be applied to well-known non-smooth problems such as sparse signal recovery and compressed sensing \cite{tropp2007signal, donoho2006compressed}, it allows the use of non-smooth sparsity-promoting regularizers. In the context of WSNs, such regularizers can for example be used to perform channel selection, and hence reduce both bandwidth and computational stress on the sensor nodes.  
 
%The paper is organized as follows. Section \ref{sec:problem_statement} describes the scope of the algorithm as well as the family of problems it applies to. Section \ref{sec:algorithm} describes the iterative optimization procedure defined by DASF for the specific case of non-smooth problems. Section \ref{sec:convergence} establishes our main convergence results along with a short proof outline. The convergence of the algorithm is validated in Section \ref{sec:examples}. Finally, we discuss some limitations of our current approach in Section \ref{sec:discussion}.

\section{Problem Statement}
\label{sec:problem_statement}

We consider a network consisting of $K$ sensor nodes, where each node $k$ collects discrete observations of an $M_k$-channel signal $\bm{y}_k(t)$. We denote $\bm{y}(t)=[\bm{y}^T_1(t),\dots,\bm{y}^T_K(t) ]^T$ the network-wide multi-channel sensor signal, where each observation is an element of $\R^M$ with $M=\sum_k M_k$.% We also define the $M\times N$ sample observation matrix $\bm{Y}(t)=[\bm{Y}^T_1(t),\dots,\bm{Y}^T_K(t) ]^T$, where $\bm{Y}_k(t)\in \R^{M_k\times N}$  $N$ consecutive observations of $\bm{y}_k(t)$ centered around sample time $t$. 
Our goal is to design a network-wide spatial filter $\bm{X}\in\R^{M\times Q}$ which fuses all the channels of $\bm{y}(t)$ into $Q$ output channels that satisfy a certain optimality criterion, which in generic form can be written as
\begin{equation}
        \label{eq:global_problem}
        \begin{split}
                \bm{X}^\star(t)\in &\argmin_{\bm{X}} f(\bm{X}^T\bm{y}(t) ,\bm{X}^T\bm{B}) + g(\bm{X}^T\bm{\Gamma})\\
                                   &\st [\bm{X}_k^T\bm{y}_k(t), \bm{X}^T_k\bm{B}_k]\in\mc{X}_k\quad \forall k.
        \end{split}
\end{equation}
%where and  are constant (time independent) matrices. %  $\mc{X}_k$ are closed subsets of $\mc{H}^{Q}$, $f:\mc{H}^Q\times \R^{Q\times D}\to \R$ is such that the function $\bm{X}\mapsto f(\bm{X}^T\bm{y}(t))$ is smooth (i.e. it is differentiable and its gradient is continuous) with compact sublevel sets, and $g:\R^{Q\times N}\to \R$ is a proper, lower semicontinuous and convex function. 
where each block $\bm{X}_k\in\R^{M_k\times Q}$ is defined according to the partitioning $\bm{X}=[\bm{X}_1^T,\cdots, \bm{X}_K^T]^T$. $\bm{\Gamma}=\text{BlkDiag}(\bm{\Gamma}_1,\dots, \bm{\Gamma}_K)$, where $\text{BlkDiag}(\cdot)$ is the operator producing a block diagonal matrix whose blocks correspond to the operator's ordered arguments, with  $\bm{\Gamma}_k\in\R^{M_k\times L_k}$ a constant (time-independent) data matrix, along with $\bm{B}=[\bm{B}_1^T, \dots, \bm{B}_K]^T\in\R^{M\times D}$. The term involving $f$ is a smooth function of $\bm{X}$ (i.e. differentiable with continuous gradient) and the term involving $g$ is a convex, possibly non-smooth, function of $\bm{X}$. The other particular characteristics, of $f$, $g$ and $\mc{X}_k$ are not immediately relevant and will thus be later described in Section \ref{sec:convergence}, along with our convergence analysis. Additionally, we require that there exist some functions $g_k$, such that the non-smooth term $g$ can be separated as
\begin{align}
        g(\bm{X}^T\bm{\Gamma}) &= \sum_k g_k(\bm{X}_k^T\bm{\Gamma}_k).
\end{align}
which holds for, e.g., the $l_1$-norm.
Typically, $f$ is a cost function depending on the second order statistics of $\bm{y}(t)$ (e.g. the covariance matrix), $g$ is a regularizing term promoting certain desirable properties of the solution (e.g. a sparsity inducing norm), and the constraint sets $\mc{X}_k$ encode some hard limits on the filter, such as limiting the maximum output power, or requiring the per-node filters to have uncorrelated outputs.  Note that, as it is defined, $g$ cannot be an indicator function, and that all the constraints must therefore be encoded in $\mc{X}_k$\footnote{This does not restrict the set of allowable problems, but allows us to simplify the notation used in Section \ref{sec:convergence}.}. As an example, a possibility for $f$ is  
\begin{equation}
        \label{eq:example}
        f(\bm{X}^T\bm{y}(t))=\E{\norm{\bm{X}^T\bm{y}(t)-\bm{d}(t)}_2^2},
\end{equation}
\looseness=-1
where $\E{\cdot}$ denotes the expectation operator, and both $\bm{y}(t)$ and $\bm{d}(t)$ are random signals. Finally, we emphasize that the per-node (i.e. per-block) constraints in \eqref{eq:global_problem} are stricter than in the original DASF problem setting in \cite{musluoglu2022unified_p1}, which allowed for coupling constraints between the different blocks of variables, i.e. between the variables of different nodes. Note that although \eqref{eq:global_problem} is only allowed to depend on $\bm{y}(t)$, it does not preclude the existence of multiple sets of signals. Indeed, by imposing the proper structure on both $f$ and $\bm{y}(t)$, we can describe problems depending on multiple sets of data. We may for example wish to solve problems of the form 
%\begin{equation}
%        \min_{\bm{X}} f_1(\bm{X}^T\bm{u}(t))+ f_2(\bm{X}^T\bm{v}(t)),
%\end{equation}
%which is possible in the framework of \eqref{eq:global_problem} by defining $\bm{y}(t) = [\bm{u}(t), \bm{v}(t)]$. We could also solve problems of the general form 
\begin{equation}
        \min_{\bm{X}_1, \bm{X}_2} f(\bm{X}_1^T\bm{u}(t),\bm{X}_2^T\bm{v}(t))
\end{equation}
which is possible in the framework of \eqref{eq:global_problem} by defining $\bm{X} = [\bm{X}_1^T, \bm{X}_2^T]^T$ and $\bm{y}(t)=\text{BlkDiag}(\bm{u}(t), \bm{v}(t))$.

For the rest of this paper, we will omit the time index $t$ of $\bm{X}^\star$, as we assume for mathematical tractability that $\bm{y}(t)$ is short-term stationary, and hence that the set of optimal filters varies slowly with time (i.e. $\bm{X}^\star(t) \approx \bm{X}^\star(t+\tau)$ for small enough $\tau$). Furthermore, we do not have access to the data-generating process $\bm{y}(t)$, but only to consecutive realizations of $\bm{y}(t)$, which, under the assumption of ergodicity, can be used to obtain an estimate of the statistics implicitly involved in \eqref{eq:global_problem}. In an actual implementation, one would evaluate/optimize \eqref{eq:global_problem} based on estimated statistics of $\bm{y}(t)$, i.e. $\bm{y}(t)$ would need to be replaced by a matrix of discrete samples $\bm{Y}(t)$ centered around $t$ and the expectation in \eqref{eq:example} would have to be approximated with a sample average.

% and that for a large enough number of samples $N$, the statistics of $\bm{y}(t)$ can be reliyably estimated. This also implies that the spatial properties of the data generating process $\bm{y}(t)$ change slowly compared to the convergence time of the iterative method described hereafter. 

%by defining $\bm{X} = [\bm{X}_1^T, \bm{X}_2^T]^T$ and $\bm{Y}=\text{BlkDiag}(\bm{u}(t), \bm{v}(t))$.

Our objective is to solve \eqref{eq:global_problem} in a bandwidth-efficient manner. The optimization procedure therefore cannot rely on a fusion center to collect samples of the full $\bm{y}(t)$ vector to estimate inter-channel statistics, as this would incur significant communication costs. Indeed, in an adaptive setting where the data is allowed to change at every iteration, every new sample would need to be collected by the FC. Instead, we propose a fully distributed procedure that relies on the nodes sharing linearly compressed views of their observations with one another, and locally solving lower dimensional versions of \eqref{eq:global_problem} at different times instance, depending only on the compressed observations received from other nodes. By exploiting the short-term stationarity of $\bm{y}(t)$, each iteration of the algorithm can be performed over a different time-window, thereby behaving like an adaptive filter in which the filter coefficients are adjusted every time a new (block of) sample(s) is collected.

\section{Non-Smooth DASF}
\label{sec:algorithm}

In order to ease the exposition of the algorithm, we limit our description to the specific case of fully-connected networks. A generalization to arbitrary topologies can be done in a similar fashion as for the original DASF algorithm \cite{musluoglu2022unified_p1, musluoglu2022unified_p2}. Furthermore, without loss of generality, we ignore the deterministic argument $\bm{X}^T \bm{B}$ as it adds a lot of clutter in the equations, while it is largely treated in the same way as the $\bm{X}^T\bm{y}(t)$ argument (we again refer to \cite{musluoglu2022unified_p1, musluoglu2022unified_p2} for further details).

In our algorithm, each node $k$ is responsible for updating its own block $\bm{X}_k\in\R^{M_k\times Q}$ of $\bm{X}=[\bm{X}_1^T,\dots, \bm{X}_K^T]^T$, corresponding to its own locally observed data $\bm{y}_k(t)$. Let us denote $\bm{X}^i$ the algorithm's estimate of the solution of \eqref{eq:global_problem} at iteration $i$. We emphasize that each iteration is performed on a different block of $N$ samples of $\bm{y}(t)$, i.e., the update from $\bm{X}^i$ to $\bm{X}^{i+1}$ will be based on the observations of $\bm{y}(t)$ at sample times $t=(i-1)N,..., iN-1$.

Let us consider problem \eqref{eq:global_problem} with the additional linear constraints 
\begin{equation}
        \label{eq:subspace_constrains}
        \bm{X}_k \in \colspace{\bm{X}^i_k}\quad \forall k\neq q,
\end{equation}
with $\colspace{\cdot}$ denoting the column space of its argument and where $q$ is some arbitrary node, which we will refer to as the updating node. By introducing the parametrization $\bm{X}_k=\bm{X}^i_k\bm{G}_k$ for $k\neq q$ corresponding to the linear subspace constraints \eqref{eq:subspace_constrains}, and defining the compressed signals of node $k$ as $\bm{z}_k^i(t)\triangleq \bm{X}_k^{iT}\bm{y}_k(t)$ and $\bm{F}_k^i\triangleq \bm{X}_k^{iT}\bm{\Gamma}_k$, the new problem \eqref{eq:global_problem} equipped with \eqref{eq:subspace_constrains} can be reformulated as
\begin{subequations}
        \label{eq:local_problem}
        \begin{align}
                \bar{\bm{X}}^\star\in &\argmin_{\bar{\bm{X}}} f(\bar{\bm{X}}^T\bm{z}^i(t)) + g(\bar{\bm{X}}^T\bm{F}^i)\\
                \st & {\bm{G}}_k^T\bm{z}_k^i(t)\in\mc{X}_k\quad \forall k \neq q \\
                    & \bm{X}^T_q\bm{y}_q(t) \in \mc{X}_q\\
                    & \bar{\bm{X}} = [\bm{G}_1^T,\dots, \bm{X}_q^T,\dots,\bm{G}_K^T]^T\label{eq:partitioning}\\
                    & \bm{z}^i(t) = [\bm{z}_1^T(t),\dots, \bm{y}_q^T(t),\dots,\bm{z}_K^T(t)]^T\\
                    & \bm{F}^i = \text{BlkDiag}(\bm{X}^{iT}_1\bm{\Gamma}_{1}, \dots, \bm{\Gamma}_{q}, \dots, \bm{X}^{iT}_K\bm{\Gamma}_{K}).
        \end{align}
\end{subequations}
We can see that by collecting the compressed observations of every other node, some node $q$ can compute a solution of the \emph{local} problem \eqref{eq:local_problem}, and equivalently of the linearly constrained \emph{global} problem \eqref{eq:global_problem} with the addition of the constraints \eqref{eq:subspace_constrains}. We use the term \emph{compressed} observations since, if $Q<M_k$, $\bm{z}_k^i(t)$ will have a lower dimension than $\bm{y}_k(t)$ and can therefore be more efficiently transmitted than the raw data. As $\bm{\Gamma}$ is assumed static, it only needs to be shared once, and only the $\bm{X}_k^i$ will need to be exchanged, unless $L<M_k$, in which case it is more efficient to share $\bm{F}^i_k$.

Note that the global and local problems \eqref{eq:global_problem} and \eqref{eq:local_problem}  have the same general structure but with a different dimension, therefore if a solver exists for the global problem, it can also be used to solve the local problems. In other words, if we denote by $\mathbb{P}(\bm{y}(t), \bm{\Gamma})$ a particular instance of problem \eqref{eq:global_problem}, solving \eqref{eq:local_problem} is equivalent to solving $\mathbb{P}(\bm{z}^i(t),\bm{F}^i)$.

Our iterative procedure consists in updating $\bm{X}^i$ by iteratively solving \eqref{eq:local_problem}, each time selecting a new node $q$ to act as the ``updating node'' in a round-robin fashion. Formally, the procedure is as follows:

\begin{enumerate}[nosep, left=5pt]
        \item{\textit{Data collection:}
                        Every node collects discrete $N$ new observations of $\bm{y}_k(t)$.
                }
        \item{\textit{Aggregation:}
                        Every node except the updating node $q$, computes its compressed data $\bm{z}_k^i(t)$ and $\bm{F}_k^i$ and transmits the corresponding $N$ compressed samples to the updating node $q$.
                }
        \item{\textit{Local solution:}
                        Based on the received compressed samples of $\bm{z}_k^i(t)$, and its own data $\bm{y}_q(t)$, the updating node $q$ can estimate the signal statistics involved in \eqref{eq:local_problem} and solve it using any solver for $\mathbb{P}(\cdot,\cdot)$. It then updates its local block as $\bm{X}^{i+1}_q = \bar{\bm{X}}^\star_q$, and extracts the optimal update matrices $\bm{G}^\star_k$ from $\bar{\bm{X}}^\star$ using the partitioning \eqref{eq:partitioning}\footnote{In the case where the local problem would have multiple solutions, the solution with the smallest distance to $\bar{\bm{X}}^{i-1}\triangleq [\bm{I},\dots,\bm{X}^{{i-1}T}_q,\dots,\bm{I}]^T$ is selected \cite{musluoglu2022unified_p1}}.
                }
        \item{\textit{Solution update:}
                        The updating node transmits the update matrices $\bm{G}^\star_k$ to their corresponding nodes. Each node except the updating node updates its block of the estimate of the solution as $\bm{X}^{i+1}_k = \bm{X}^i_k\bm{G}^\star_k$.
                }
\end{enumerate}
The full description of the algorithm is given by Algorithm \ref{alg:dasf}, which we refer to as non-smooth DASF (NS-DASF).

\SetKwFor{Node}{At node}{}{}
\SetKwFor{Loop}{loop}{}{}
\begin{algorithm}[t]
        \Begin{
                $i\gets 0$, $q\gets 1$, Randomly initialize $\bm{X}^0$\\
                \Loop{}{
                        \For{$k\in\{1,\dots, K\}\smallsetminus\{q\}$}{
                                \Node{$k$} {
                                        Collect a new batch of $N$ samples of $\bm{y}_k(t)$ and send the compressed samples $\bm{z}_k^i(t)=\bm{X}_k^{iT}\bm{y}_k(t)$ along with $\bm{F}_k^i=\bm{X}_k^{iT}\bm{\Gamma}_k$ to node $q$.
                                }
                        }
                        \Node{$q$}{
                                Obtain $\bar{\bm{X}}^\star$ by solving the local problem \eqref{eq:local_problem} using only the compressed data $\bm{z}^i(t)$ and $\bm{\Gamma}^i$. If the solution is not unique, select the one minimizing $\norm{\bar{\bm{X}}^\star-\bar{\bm{X}}^{i-1}}_F$.\\
                                Extract $\bm{X}_q^\star$ and the $\bm{G}_k^\star$'s  from $\bar{\bm{X}}^\star$ according to the partitioning \eqref{eq:partitioning}.\\
                                $\bm{X}_q^{i+1}\gets \bm{X}_q^\star$\\
                                \For{$k\in\mathcal{K}\smallsetminus\{q\}$}{

                                        Send $\bm{G}_k^\star$ to node $k$.\\
                                        \Node{$k$}{
                                                $\bm{X}^{i+1}_k \gets \bm{X}^{i}_k\bm{G}_k^\star$\\
                                        }
                                }

                        }
                        $i\gets i+1$, $q\gets (q \mod K) + 1$

                }
        }
        \caption{NS-DASF algorithm.}
        \label{alg:dasf}
\end{algorithm}
\vspace{-15pt}

\section{Convergence}
One can gain intuition about the algorithm's convergence  by noting that $\bm{X}^{i}$ is always in the feasible set of problem \eqref{eq:local_problem}, as it satisfies \eqref{eq:subspace_constrains} trivially, ensuring a monotonic decrease of the objective.
\label{sec:convergence}
%The convergence of the algorithm in the smooth case (i.e. $g(\cdot)=0$) has been fully proven in \cite{musluoglu2022unified_p2}, but convergence  to a stationary point of \eqref{eq:global_problem} in the non-smooth  case remains to be proven. Let us rephrase the main result of \cite{musluoglu2022unified_p2}:
%\begin{theorem}
%        Let $(\bm{X}^i)_{i\in\mathbb{N}}$ denote a sequence of iterates generated by Algorithm \ref{alg:dasf} and assume that the solution set of \eqref{eq:global_problem} is non-empty and varies continuously\footnote{Continuity must here be understood in the context of point-to-set maps. More specifically, we require \emph{upper hemicontinuity}. For details see  \cite{Berge1997, Charalambos2013}.} with the data $\bm{Y}$. Then every accumulation point of  $(\bm{X}^i)_{i\in\mathbb{N}}$ is a fixed point of the algorithm.
%\end{theorem}
We will start by showing that fixed points of Algorithm \ref{alg:dasf} (i.e. points $\bm{X}^*$ such that if $\bm{X}^0=\bm{X}^*$, then $(\bm{X}^i)_{i\in\mathbb{N}}=(\bm{X}^*)_{i\in\mathbb{N}}$) are stationary points of problem \eqref{eq:global_problem}, and then reuse one of the result of \cite{musluoglu2022unified_p2} to show convergence to such a point. Let us first define
\begin{equation}
        \begin{split}
                p(\bm{X})&\triangleq f(\bm{X}^T\bm{y}(t))\\
                q(\bm{X})&\triangleq g(\bm{X}^T\bm{\Gamma})\\
                q_k(\bm{X}_k)&\triangleq g_k(\bm{X}_k^T\bm{\Gamma}_k)
        \end{split}
        \quad 
        \begin{split}
                \mc{D}_k&\triangleq \{\bm{X}_k\;|\; \bm{X}_k^T\bm{y}_k(t)\in\mc{X}_k\}\\
                \mc{D}&\triangleq\mc{D}_1\times \cdots\times \mc{D}_K
        \end{split}
\end{equation}
%$q(\bm{X})\triangleq f(\bm{X}^T\bm{Y})$, $p(\bm{X})\triangleq g(\bm{X}^T\bm{\Gamma})$, $p_k(\bm{X}_k)\triangleq g_k(\bm{X}_k^T\bm{\Gamma}_k)$, $\mc{D}_k\triangleq \{\bm{X}\;|\; \bm{X}_k^T\bm{Y}_k\in\mc{X}_k\}$ and $\mc{D}=\mc{D}_1\times \cdots\times \mc{D}_K$. 
where $\times$ denotes the cartesian product between sets. We assume that $p:\R^{M\times Q}\to \R$ is a smooth function with compact sublevel sets, $q:\R^{M\times Q}\to\R$ is a proper, lower semicontinuous and convex function, and $\mc{D}$ is a closed set in $\R^{M\times Q}$. We wish to show that the fixed points of the algorithm are also stationary points of problem \eqref{eq:global_problem}, that is feasible points $\bm{X}^\star$ such that \cite{rockafellar2009variational}
\begin{equation}
        \label{eq:stationarity_conditions}
        0\in\nabla p(\bm{X}^\star)+ \partial q(\bm{X}^\star) + N_\mc{D}(\bm{X}^\star),
\end{equation}
where $\partial q(\cdot)$ denotes the set of subgradients of $q$ at a particular point and $N_{(\cdot)}(\cdot)$ denotes the normal cone at a particular point of a set. The sum  between sets must be interpreted as a Minkowski sum\footnote{$A+B=\{a+b\;|\; a\in A,\; b\in B\}$.}. Equation \eqref{eq:stationarity_conditions} generalizes the well-known Karush-Kuhn-Tucker (KKT) conditions \cite{karush1939minima ,kuhnnonlinear} to the case of non-smooth functions\footnote{\cite{li2020understanding} contains a useful introduction to the concepts of stationnarity for non-smooth problems.} (it therefore reduces to the KKT conditions in the smooth case).
% Under the qualification
%\begin{multline}
%        \bm{W}\in N_{\text{dom } p}(\bm{X}) \text{ and } -\bm{W}\in N_{\mc{D}}(\bm{X}) \Rightarrow \bm{W} = 0,
%\end{multline}
It merely gives necessary conditions for a feasible point to be a solution of \eqref{eq:global_problem}, but the condition is also sufficient in the case of convex instances of the problem \cite{rockafellar2009variational}. Intuitively, those points are such that all directional derivatives pointing inside the feasible set are positive (i.e. there is no feasible descent direction at that point, see \cite{rockafellar2009variational, Royset2021} for details). 

Before stating our main result, we give an explicit expression of the normal cone corresponding to the subspace constraints \eqref{eq:subspace_constrains} at a point $\bm{X}^i=\bm{X}$. We denote 
\begin{equation}
        \mc{L}_q(\bm{X})\triangleq \colspace{\bm{X}_1} \times \cdots \times \R^{ M_q\times Q} \times \cdots \times \colspace{\bm{X}_K}
\end{equation}
the subspace constraints at node $q$, where $\bm{X}$ here corresponds to $\bm{X}^i$ in \eqref{eq:subspace_constrains} and where  $\R^{ M_q\times Q}$ corresponds to the lack of constraints associated with node $q$. As the normal cone to a linear subspace is simply its orthogonal complement \cite{Royset2021}, we have
\begin{equation}
        \label{eq:def_local_normal_cone}
        N_k(\bm{X})\triangleq N_{\mc{L}_q(\bm{X})}(\bm{X}) = \colspace{\bm{X}_1}^\perp \times \cdots \times\{0\} \times \cdots \times\colspace{\bm{X}_K}^\perp,
\end{equation}
where $(\cdot)^\perp$ denotes the orthogonal complement and the singleton $\{0\}=(\R^{M_q\times Q})^\perp$. We can now state a first result, which established the optimality of fixed points of Algorithm \ref{alg:dasf} under a mild technical condition which is akin to the well-known linear independence constraint qualification (LICQ). 
\begin{theorem}
        \label{theo:fixed_are_stat}
        Let $\bm{X}^*$ be a fixed point of Algorithm \ref{alg:dasf} and assume that the following constraint qualifications hold:
        \begin{equation}
                \label{eq:qualification}
                N_\mc{D}(\bm{X}^*)\cap N_k(\bm{X}^*) = \{0\}\quad \forall k.
        \end{equation}
        Then $\bm{X}^*$ satisfies the stationary conditions \eqref{eq:stationarity_conditions} and is therefore a stationary point of problem \eqref{eq:global_problem}.
\end{theorem}
\begin{proof}
        The qualification \eqref{eq:qualification} can be viewed as a generalization of the traditional LICQ \cite{Royset2021}, and ensures that the solutions of the local problems \eqref{eq:local_problem} satisfy\footnote{This is true in part because all the properties of $f$, $g$, and $\mc{X}_k$ described at the beginning of Section \ref{sec:problem_statement} are inherited by $p,q$ and $\mc{D}_k$. We omit this part of the proof due to the page limit.} the stationarity conditions of \eqref{eq:local_problem} (or equivalently \eqref{eq:global_problem} with the additional constraints \eqref{eq:subspace_constrains}) \cite{rockafellar2009variational}
        \begin{equation}
                0\in\nabla p(\bm{X}^*)+ \partial q(\bm{X}^*) + N_\mc{D}(\bm{X}^*)+N_k(\bm{X}^*)\quad \forall k,
        \end{equation}
        or equivalently
        \begin{equation}
                \begin{split}
                        \forall k, \exists z_k \in \partial q(\bm{X}^*)+N_\mc{D}(\bm{X}^*), \exists a_k \in  N_k(\bm{X}^*): \\ \nabla p(\bm{X}^*)+z_k+a_k = 0.
                \end{split}
        \end{equation}
        Let $a_k^k$ and $z_k^k$ denote the blocks corresponding to node $k$ within $a_k$ and $z_k$, respectively. Similarly, Let $\nabla_k p(\bm{X}^*)$ correspond to the block of the gradient associated with the block $\bm{X}_k^*$. Then 
        $%\begin{equation}
                0 = z^k_k + \nabla_k p(\bm{X}^*),\;
        $%\end{equation}
        as $a_k^k \in \{0\}$ from the definition \eqref{eq:def_local_normal_cone}. From the block separability of $\mc{D}$ and $q$, we have that \cite{rockafellar2009variational}
        \begin{subequations}
                \begin{align}
                        N_\mc{D}&=N_{\mc{D}_1} \times \cdots \times N_{\mc{D}_K}\\
                        \partial q(\bm{X})&=\partial q_1(\bm{X}) \times \cdots \times\partial q_K(\bm{X}).
                \end{align}
        \end{subequations}
        Therefore it must be that
        \begin{equation}
                z_k^k=-\nabla_k p(\bm{X}^*)\in\partial q_k(\bm{X}^*)+N_{\mc{D}_k}(\bm{X}^*)\quad \forall k,
        \end{equation}
        and therefore
        $%\begin{equation}
                -\nabla p(\bm{X}^*)\in\partial q(\bm{X}^*)+N_{\mc{D}}(\bm{X}^*),\;
        $%\end{equation}
         i.e. \eqref{eq:stationarity_conditions} is satisfied.
\end{proof}
In the case where the constraint set consists of smooth equality and inequality constraints, we have the following corollary.
\begin{corrolary} (Proof omitted)
        Let $\bm{X}^*$ be a fixed point of Algorithm \ref{alg:dasf} and let $u^k_j:\R^{M\times Q}\to\R, v^k_l:\R^{M\times Q}\to\R$ be smooth functions $\forall j,l,k$. If the constraint sets $\mc{D}_k$ can be expressed as 
        \begin{equation}
                \mc{D}_k=\{\bm{X}_k\;|\; u^k_j(\bm{X})=0, v^k_l(\bm{X})\leq 0\; \forall j,l\}
        \end{equation} 
        and it holds that the element of the set 
        \begin{equation}
                \label{eq:block_LICQ}
        \{{\bm{X}^*_k}^T\nabla u^k_j(\bm{X}_k^*)\;\forall j; \;{\bm{X}^*_k}^T\nabla v^k_l(\bm{X}_k^*)\;\forall l\in \mathbb{A}(\bm{X}^*)\},
        \end{equation}
        where $\mathbb{A}(\bm{X}^*)$ denotes the active inequality constraints at $\bm{X}^*$, are linearly independent for every $k$, then the qualification \eqref{eq:qualification} is satisfied and $\bm{X}^*$ is a stationary point of problem \eqref{eq:global_problem}.
\end{corrolary}
\noindent The qualification \eqref{eq:block_LICQ} can be seen as a stricter
version of the well-known LICQ, where each of the blocks of the gradients are
required to be independent when projected on the column-spaces of the blocks of ${\bm{X}^*}$, instead of the gradients themselves. 

We will now rephrase \cite[Theorem 6]{musluoglu2022unified_p2}, which asserts convergence of Algorithm \ref{alg:dasf} (the proof is the same as in \cite{musluoglu2022unified_p2} since it does not depend on the (non-)smoothness of the objective, except for the part associated with Theorem \ref{theo:fixed_are_stat}, which was proven above).
\begin{theorem}
        Let $(\bm{X}^i)_{i\in\mathbb{N}}$ denote a sequence of iterates generated by Algorithm \ref{alg:dasf} and assume that the solution set of \eqref{eq:global_problem} is non-empty and varies continuously\footnote{Continuity must here be understood in the context of point-to-set maps. More specifically, we require \emph{upper hemicontinuity}. For details see  \cite{Berge1997, Charalambos2013}.} with the problem's parameters $\bm{y}(t)$ and $\bm{\Gamma}$. Furthermore, assume that the number of stationary points of \eqref{eq:global_problem} is finite (or the number of reachable stationary points of the solver of \eqref{eq:local_problem} is finite). Then $(\bm{X}^i)_{i\in\mathbb{N}}$ converges to a stationary point of problem \eqref{eq:global_problem}.
\end{theorem}

\vspace{-5pt}
\section{Simulated Example}
\vspace{-5pt}
\label{sec:examples}
Consider the sparse multichannel Wiener filtering problem 
\begin{equation}
        \label{eq:sparse_recovery_problem}
        \min_{\bm{X}} \E{\norm{\bm{X}^T\bm{y}(t)-\bm{d}(t)}_2^2}+\lambda  \norm{\bm{X}}_1,
\end{equation}
where $\bm{d}(t)$ is some desired $Q$-channel filter output signal. For the following simulations, we generated instances of the problem as $\bm{d}(t) = \bm{X}^{\star T} \bm{y}(t) + \bm{n}(t)$, where the entries of $\bm{y}(t)$ and $\bm{n}(t)$ are i.i.d. zero-mean random gaussian signals with variance $1$ and $0.1$, respectively. $\bm{X}^*$ is an $(\frac{M}{10})$-sparse random vector with zero-mean and unit variance gaussian entries. Furthermore, we set $\lambda=1$, $Q=1$,  $K=10$, $M_k=10$,  and generate a 1000 samples of $\bm{y}(t)$, $\bm{d}(t)$ and $\bm{n}(t)$ for each experiment. The expectation in \eqref{eq:sparse_recovery_problem} is computed as a simple sample average. The local version of \eqref{eq:sparse_recovery_problem} was solved using Chambolle-Pock's algorithm \cite{chambolle2011first}, and we therefore only approximate the optimal local solution of \eqref{eq:local_problem}.

%\paragraph*{Static case}
For the case of a problem which does not vary in time, we performed a Monte Carlo simulation consisting of 100 runs, with the parameters described above. Different $\bm{y}(t)$, $\bm{X}^\star$ and $\bm{n}(t)$ were randomly generated for each run. Figure \ref{fig:static} depicts the convergence in terms of the relative mean-squared-error ${\norm{\bm{X}^i-\bm{X}^\star}_F^2}/{\norm{\bm{X}^\star}_F^2}$.
We see that the algorithm consistently converges to reasonable accuracy within
two full rounds (i.e. each node has solved the local problem twice,
corresponding to 20
iterations in our example). The remaining static error should be attributed to the error inherent to the iterative method used to solve the local problems, and not to our procedure itself (as implied by Theorem \ref{theo:fixed_are_stat}).

\begin{figure}[t]
        \includegraphics[width=\linewidth]{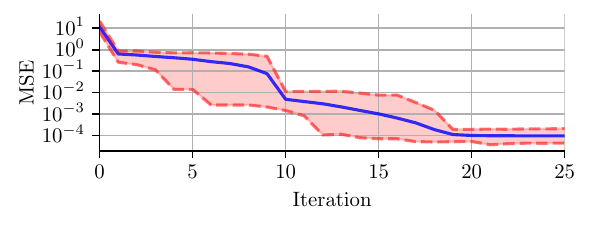}
        \vspace{-22pt}
\caption{Convergence of Algorithm \ref{alg:dasf} applied to problem \eqref{eq:sparse_recovery_problem}. Dashed red curves correspond to the min-max convergence curves. The blue curve corresponds to the median convergence curve.}
\label{fig:static}
\end{figure}

%\paragraph*{Dynamic case}
Although we do not provide any proof or quantitative relationship between the rate of change of $\bm{X}^\star(t)$ and the relative error of the algorithm's estimate of the solution, we illustrate the tracking capabilities of the algorithm with a particular example depicted in Figure \ref{fig:dynamic}. Two sparse vectors $\bm{X}_A$ and $\bm{X}_B$ were drawn from the same distribution used for $\bm{X}^\star$ in the static case, and $\bm{X}^\star(t)$ was computed as $w(t)\bm{X}_A+(1-w(t))\bm{X}_B$, where $w(t) = t\cos(t^4)$. The time at iteration $i$ is $t_i=i/180$. Figure \ref{fig:dynamic} depicts the projection of the optimal solution and the algorithm's estimate on the line joining $\bm{X}_A$ to $\bm{X}_B$. We see that as the rate of change of the optimal solution increases, the algorithm starts lagging behind the optimal solution.

\begin{figure}[t]
        \includegraphics[width=\linewidth]{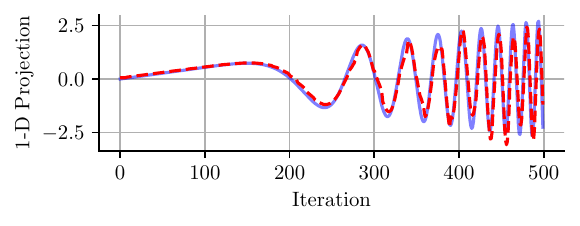}
        \vspace{-22pt}
        \caption{Tracking of an accelerating solution over time. The blue curve corresponds to the optimal solution, the red dashed curve corresponds to NS-DASF's estimate.}
        \label{fig:dynamic}
\end{figure}

\vspace{-5pt}
\section{Conclusion and Future Work}
\vspace{-5pt}
\label{sec:discussion}
In this paper, we have described a distributed adaptive algorithm to solve a particular family of non-smooth spatial filtering problems. The algorithm was validated both by a formal proof and numerical simulations. In future works, we will provide an analysis of the convergence properties of the algorithm and investigate the link between the global solution accuracy, the local accuracy, and the rate of change of the data (i.e. the tracking performance of the algorithm).% Secondly, the convergence results presented in this paper somewhat limits the original scope of DASF, as block separability of the constraints is required. A stronger result allowing non-separable constraints exists, but it relies on a qualification involving both the constraint set and the subdifferential of the non-smooth-term. The implications of this qualification will be studied in future work.

\bibliographystyle{IEEEtran}
\bibliography{main_biblio}

\end{document}